\documentclass[%
 superscriptaddress,
 onecolumn,
 showpacs,
 showkeys,
 nofootinbib,
  amsmath,amssymb,
  aps,
  longbibliography,
  floatfix,
 ]{revtex4-1}

\usepackage[normalem]{ulem}

\usepackage{adjustbox}

\usepackage{hyperref}
\usepackage{amsmath}
\usepackage{amssymb}
\usepackage{amsthm}
\usepackage{bm} 
\usepackage{graphicx}

\RequirePackage{times}
\RequirePackage{mathptm}

\usepackage{url}
\usepackage[x11names]{xcolor}
\usepackage{eepic}
\usepackage{tikz}
\usepackage{epstopdf}
\usepackage[normalem]{ulem}
\newtheorem{theorem}{Theorem}

\newtheorem{lemma}{Lemma}
\theoremstyle{definition}
\newtheorem{definition}{Definition}

\newcommand{\R}{\mathbb{R}}
\newcommand{\C}{\mathbb{C}}

\newcommand{\dd}{\mathrm{d}}

\newcommand{\bra}[1]{\left< #1 \right|}
\newcommand{\ket}[1]{\left| #1 \right>}

\newcommand{\iprod}[2]{\langle #1 | #2 \rangle}

\newcommand{\oprod}[2]{| #1 \rangle\langle #2 |}

\begin{document}

\title{A variant of the Kochen-Specker theorem localising value indefiniteness}
	
\author{Alastair A. Abbott}
\email{a.abbott@auckland.ac.nz}
\homepage{http://www.cs.auckland.ac.nz/~aabb009}

\affiliation{Department of Computer Science, University of Auckland,
Private Bag 92019, Auckland, New Zealand}
\affiliation{Centre Cavaill\`es, \'Ecole Normale Sup\'erieure, 29 rue d'Ulm, 75005 Paris, France}

\author{Cristian S. Calude}
\email{cristian@cs.auckland.ac.nz}
\homepage{http://www.cs.auckland.ac.nz/~cristian}

\affiliation{Department of Computer Science, University of Auckland,
Private Bag 92019, Auckland, New Zealand}

\author{Karl Svozil}
\email{svozil@tuwien.ac.at}
\homepage{http://tph.tuwien.ac.at/~svozil}

\affiliation{Institute for Theoretical Physics,
Vienna  University of Technology,
Wiedner Hauptstrasse 8-10/136,
1040 Vienna,  Austria}

\affiliation{Department of Computer Science, University of Auckland,
Private Bag 92019, Auckland, New Zealand}

\date{\today}

\begin{abstract}
The Kochen-Specker theorem proves the inability to assign, simultaneously, noncontextual definite values to all (of a finite set of) quantum mechanical observables in a consistent manner.
If one assumes that any definite values behave noncontextually, one can nonetheless only conclude that \emph{some} observables (in this set) are value indefinite. 


In this paper we prove a variant of the Kochen-Specker theorem showing that, under the same assumption of noncontextuality, if a single one-dimensional projection observable is assigned the definite value 1, then no one-dimensional projection observable that is incompatible (i.e., non-commuting) with this one can be assigned consistently a definite value.
Unlike standard proofs of the Kochen-Specker theorem, in order to localise and show the extent of value indefiniteness this result requires a constructive method of reduction between Kochen-Specker sets.

If a system is prepared in a pure state $\ket{\psi}$, then it is reasonable to assume that any value assignment (i.e., hidden variable model) for this system assigns the value 1 to the observable projecting onto the one-dimensional linear subspace spanned by $\ket{\psi}$, and the value 0 to those projecting onto linear subspaces orthogonal to it.
Our result can be interpreted, under this assumption, as showing that the outcome of a measurement of \emph{any} other incompatible one-dimensional projection observable cannot be determined in advance, thus formalising a notion of quantum randomness.


\end{abstract}


\maketitle

\section{The Kochen-Specker theorem and value indefiniteness}

Bell's theorem~\cite{Bell:1966uq} and the Kochen-Specker theorem~\cite{Kochen:1967fk} are perhaps two of the results which have been most influential in developing the modern understanding of quantum mechanics as an irreducibly non-classical theory~\cite{Mermin:1993qy,Pitowsky:1998aa}.
Moreover, these two no-go theorems are seen as the strongest argument for quantum mechanics being a fundamentally indeterministic theory, rather than one ruled by a deeper determinism below the level of the quantum mechanical description of reality.

Bell's theorem, which shows that quantum mechanics predicts statistical correlations between separated particles greater than what would be possible in any local, realistic, classical theory, was the focus of attention for several decades due to its relatively clear ability to be tested experimentally~\cite{Aspect:1982dp}.
The Kochen-Specker theorem was proved very shortly afterwards,
but was largely ignored due to a perceived lack of testability,
and perhaps also its formalisation in terms of partial algebras,
until it attracted renewed attention with the more recent advances
in quantum information theory and foundations.
In contrast to the bounds on probability distributions given by Bell's theorem, the Kochen-Specker theorem shows that the Hilbert-space structure of quantum mechanics makes it  impossible to assign `classical' definite values to all quantum observables in a consistent manner.
Since such a definite value is precisely a (deterministic) hidden variable specifying, in advance, the result of a measurement of an observable, this means that the outcomes of all quantum measurements on a system cannot be simultaneously predetermined.
More recent developments have significantly reduced the size and difficulty of proofs of the Kochen-Specker theorem~\cite{Cabello:1996zr} and converted such proofs into testable inequalities~\cite{Cabello:2008hc}.

However, in showing the impossibility of a classical deterministic `two-valued' measure (i.e., value assignment)
the Kochen-Specker theorem leaves open several possible conclusions.
The Kochen-Specker theorem, more specifically, finds a contradiction between the following three assumptions, which will be formalised more rigorously a little later:
\begin{itemize}
\item[(i)] all observables are assigned a definite value (i.e., are `value definite');
\item[(ii)] this definite value should be noncontextual -- that is, assigned as a function of the observable alone, and not depending on other compatible observables;
\item[(iii)] the definite values for a set of compatible observables must be consistent with the theoretical quantum predictions for the relations between them.
\end{itemize}
Condition (iii) is largely uncontroversial and hence one must generally conclude that either (or even both) (i) and (ii) must be given up.
Several alternative interpretations of quantum mechanics are contextual (e.g.,~\cite{Bohm:1952aa}), and hence discard (ii).
Perhaps the more popular interpretation, however, is that 
the inability to simultaneously assign noncontextual definite values, representing predetermined measurement outcomes, to all observables means that measurement outcomes are not determined in advance at all: that quantum mechanics represents a value indefinite reality.
This interpretation is often referred to simply as `contextuality' in the literature; however we reserve this term strictly for the contextual behaviour of definite values.

If we choose to require (ii) to hold, at least for any observables that are assigned definite values, 
then there remains an oft-overlooked gap between the formal result of the Kochen-Specker theorem and the general interpretation of it.
Indeed, the negation of (i) is that \emph{not all} observables are assigned definite values: it does not prove that no observable can be assigned a definite value and hence, given that definite values represent predetermined measurement outcomes, does not show
that all measurements must result in the \emph{ex nihilo} creation of an outcome, nor does it allow one to know \emph{which} observables in any set are value indefinite.
We can, of course, postulate that if some observables are value indefinite, then this should,
by symmetry or uniformity considerations, be the case for all observables (or at least those for which the Born rule assigns a probability strictly between 0 and 1 to some outcomes).
However, it is key to realise that this is not in any sense a formal consequence of the Kochen-Specker theorem, and constitutes an additional, undesired, assumption.

In this paper we address precisely this issue.
As is common in modern treatments of the Kochen-Specker theorem~\cite{Cabello:1994ly,Cabello:1996zr,Peres:1991ys} we focus on one-dimensional (rank-1) projection observables, and we denote the observable projecting onto the linear subspace spanned by a vector $\ket{\psi}$ as $P_\psi=\frac{\oprod{\psi}{\psi}}{|\iprod{\psi}{\psi}|}$.
By using a modified, weakened set of assumptions, we prove that if one such projection observable $P_\psi$ is assigned the value 1, then no other such projection observable $P_\phi$ can be consistently assigned a definite value unless $P_\psi$ and $P_\phi$ commute.
In interpreting this result physically, we note that if a system is prepared in the state $\ket{\psi}$ then the outcome of a measurement of the observable $P_\psi$ 
is known to be 1 with certainty, and
thus any value assignment representing the outcomes of possible measurements on the system should assign the value 1 to $P_\psi$.
If $P_\phi$ does not commute with $P_\psi$ it is therefore value indefinite under such a value assignment and hence cannot have a consistently predetermined measurement outcome.

This self-contained, analytic proof extends and generalises the results of~\cite{Abbott:2012fk,Abbott:2013ly}.

Throughout the paper we will assume (ii) to hold, as is common in interpretations of the Kochen-Specker theorem, and our strengthened results and interpretation of the Kochen-Specker theorem thus rely on this condition.
We do not attempt to justify this assumption here, as this is an interpretational choice and the subject of much debate (see~\cite[Chap. 4]{Ronde:2011aa} for an overview), which is beyond the scope of this paper.

\subsection{Definitions}
\label{defs}
As usual we denote by $\C$ the set of complex numbers and use the standard quantum mechanical bra-ket notion; that is, we denote (unit) vectors in the Hilbert space $\C^n$ by $\ket{\cdot}$.
As mentioned above, we will focus on one-dimensional projection observables and denote by $P_\psi$ the operator projecting onto the linear subspace spanned by $\ket{\psi}$; that is, $P_\psi=\frac{\ket{\psi}\bra{\psi}}{|\iprod{\psi}{\psi}|}$.

In the following we formalise hidden variables and the notion of value definiteness in a clear and unambiguous fashion.
This framework is based on that we developed in~\cite{Abbott:2012fk}, and similar to standard approaches to the Kochen-Specker theorem~\cite{Cabello:1994ly};
we have made several simplifications since we do not wish to explore contextual definite values or hidden variable theories in any detail here.

We fix a positive integer $n\ge 2$.
Let $\mathcal{O} \subseteq \{ P_\psi \mid \ket{\psi} \in \mathbb{C}^n \}$ be a nonempty set of one-dimensional projection observables on the Hilbert space $\mathbb{C}^n$~\cite{Halmos:1974qv}.
\begin{definition}
	A set $C\subset \mathcal{O}$ is a \emph{context of $\mathcal{O}$} if $C$ has $n$ elements (i.e., $|C|=n$) and for all $P_\psi,P_\phi\in C$ with $P_\psi \neq P_\phi$, $\iprod{\psi}{\phi}=0$.
\end{definition}
Since distinct one-dimensional projection observables commute if and only if they project onto mutually orthogonal linear subspaces, a context $C$ of $\mathcal{O}$ is thus a maximal set of compatible one-dimensional projection observables on $\C^n$.
Because there is a direct correspondence (up to a phase-shift) between unit vectors and one-dimensional projection observables, a context is uniquely defined by an orthonormal basis of $\C^n$.

Recall that a partial function is one which may be undefined for some values.
If it is defined everywhere, then it is total.
\begin{definition}\label{def:v}
	A \emph{value assignment function (on $\mathcal{O}$)} is a partial two-valued function $v: \mathcal{O} \to \{0,1\}$, assigning values to some (possibly all) observables in $\mathcal{O}$.
\end{definition}
We note that we could, as in~\cite{Abbott:2012fk}, allow $v$ to be a function of both the observable $P$ and the context $C$ containing $P$, allowing values to be assigned contextually.
It would perhaps be more correct to call $v$, as defined above in Definition~\ref{def:v}, a \emph{noncontextual} value assignment function; however, since we are interested only in the noncontextual case, we avoid this for compactness.

\begin{definition}
	An observable $P\in\mathcal{O}$ is \emph{value definite (under $v$)} if $v(P)$ is defined; otherwise it is \emph{value indefinite (under $v$)}.
	Similarly, we call $\mathcal{O}$ value definite (under $v$) if every observable $P\in\mathcal{O}$ is value definite.
\end{definition}

\subsection{The Kochen-Specker theorem}

With this terminology, we can state the Kochen-Specker theorem formally.
We present it in the following form deliberately in order to draw the comparison to our earlier informal description, and to clarify the following discussion, even though the second condition is redundant because we require, by definition, that a value assignment function be noncontextual.

\begin{theorem}[Kochen-Specker, \cite{Kochen:1967fk}]
	Let $n\ge 3$.
	Then there exists a (finite) set of one-dimensional projection observables $\mathcal{O}$ on the Hilbert space $\C^n$ such that there is no value assignment function $v$ satisfying the following three conditions:
	\begin{itemize}
		\item[(i)] $\mathcal{O}$ is value definite under $v$; that is, $v$ is a total function.
		\item[(ii)] The value $v(P)$ of an observable $P\in\mathcal{O}$ depends only on $P$ and not the context containing $P$.
		\item[(iii)] For every context $C$ of $\mathcal{O}$ the following condition holds%
		\footnote{This condition means that $v$ is a \emph{Boolean frame function} with weight 1~\cite{Gleason:1957ty}.}:
		$\sum_{P\in C}v(P)=1$.
	\end{itemize}
\end{theorem}

The third condition expresses the fact that only one projection observable in a context can be assigned the value 1.
As we mentioned earlier, this is largely uncontroversial: one can simultaneously measure the observables in a context and quantum mechanics predicts precisely that exactly one of these measurements should give the result `1', thus any corresponding definite values assigned to these observables should obey this same condition.
Hence, if we  assume (ii) to be true -- at least for observables that are value definite for which the statement makes clear sense -- then the Kochen-Specker theorem requires us to conclude the negation of (i): that $\mathcal{O}$ cannot be value definite, and hence \emph{at least one} observable must be value indefinite.

Note that the third condition is not independent of the first: it is not clear how the sum $\sum_{P\in C}v(P)$ should be evaluated if $v(P)$ is undefined.
This is one of the key issues we will clarify in attempting to localise value indefiniteness.

\section{A path to localising value indefiniteness}

While the Kochen-Specker theorem certainly succeeds, as was the original intention, in showing that quantum mechanics must obey an entirely non-classical event structure, it does not, as we have pointed out, show that all observables must be value indefinite and their outcomes intrinsically indeterministic.
As a consequence of the global nature of the hypothesis of the theorem -- that \emph{all} observables are value definite -- one can only draw a global conclusion: that \emph{not all} observables are value definite.
That is, the theorem, even under the assumption of noncontextuality, cannot `locate' value indefiniteness to any particular observable.
This is an important point, not only for the foundational understanding of quantum mechanics, but also in practical applications:
quantum random number generators and cryptographic schemes rely on the indeterminism of quantum mechanics providing `irreducible randomness'~\cite{Fiorentino:2007dq}.
To certify such claims, it is important to be able to localise value indefiniteness to ensure it applies to the observables measured in such applications.

We proceed by providing more nuanced and less demanding, localised versions of the Kochen-Specker assumptions, and use these to localise value indefiniteness (always under the assumption of noncontextuality for value definite observables).

\subsection{Localising the hypotheses}

Our approach is a conservative one: rather than assuming complete value definiteness of the entire set of observables considered, we require observables to be value definite only when their indefiniteness would allow the possibility of measurements\footnote{If an observable is value indefinite, this must surely imply that both outcomes are \emph{possibilities}.} violating the quantum predictions specified in condition (iii) of the Kochen-Specker theorem (see the more detailed discussion and example below).

In order for this approach to work, we need, as a premise, at least one observable to be value definite.
We then show that the assumption that \emph{any other} observable is value definite leads to a contradiction.

Fortunately, there is a justification for this premise: if a system is prepared in an arbitrary state $\ket{\psi}\in\C^n$, then measurement of the observable $P_\psi$ should yield the outcome 1. Thus, it seems reasonable to require that, if $P_\psi\in\mathcal{O}$, then $v(P_\psi)=1$.
We call this the \emph{eigenstate assumption}~\cite{Abbott:2012fk}, which is similar to, although weaker than, the `eigenstate-eigenvalue link' discussed in~\cite{Suarez:2004gn}.
Furthermore, since the critical feature of a set $\mathcal{O}$ of projection observables is the orthogonality relations between these observables rather than the specific form of these observables, we can hence choose our basis at will.
It is thus not unreasonable to consider that some one-dimensional projection observable in $\mathcal{O}$ has the value 1, and to fix the basis used to express $\mathcal{O}$ to that of the state $\ket{\psi}$ to make this observable coincide with $P_\psi$.

Let us finally discuss how assumption (iii) can be generalised for partial value assignment functions $v$, that is, the case where some observables in $\mathcal{O}$ may be value indefinite.
\begin{definition}[Admissibility]\label{admis-rules}
	Let $\mathcal{O}$ be a set of one-dimensional projection observables on $\C^n$ and let $v:\mathcal{O}\to\{0,1\}$ be a value assignment function.
	Then $v$ is \emph{admissible} if the following two conditions hold for every context $C$ of $\mathcal{O}$:
	\begin{itemize}
		\item[(a)] if there exists a $P\in C$ with $v(P)=1$, then $v(P')=0$ for all $P'\in C\setminus\{P\}$;
		\item[(b)] if there exists a $P\in C$ with $v(P')=0$ for all $P'\in C\setminus\{P\}$, then $v(P)=1$.
	\end{itemize}
\end{definition}
Admissibility requires that the quantum predictions of (iii) are never violated, while allowing the value indefiniteness of an observable $P$ if both outcomes (0 and 1) of a measurement of $P$ would be compatible with the definite values of other observables sharing a context with $P$.
For example, if $v(P)=1$, then a measurement of all the observables in a context $C$ containing $P$ must yield the outcome 1 for $P$, and hence to avoid contradiction the outcome 0 for the other observables in the context.
On the other hand, if $v(P)=0$, even though measurement of $P$ must yield the outcome 0, any of the other observables in $C$ could yield the value 1 or 0 (as long as only one yields 1), hence we should not conclude the value definiteness of these other observables.

\subsubsection{An illustrated example}
\label{sec:CabelloExample}

Let us illustrate the difference between our weakened assumptions, and in particular admissibility, with the hypotheses of the Kochen-Specker theorem.

\begin{figure}
	\begin{center}
	\begin{tabular}{ccc}
		\if01
		\begin{tikzpicture}  [scale=0.6]

		\tikzstyle{every path}=[line width=1pt]
		\tikzstyle{c1}=[circle,minimum size=6]
		\tikzstyle{s1}=[rectangle,minimum size=9]
		\tikzstyle{l1}=[draw=none,circle,minimum size=35]
		\tikzstyle{l2}=[draw=none,circle,minimum size=12]

		\path (0:5) coordinate(1)
			  (4.167,1.44) coordinate(2)
			  (3.33,2.88) coordinate(3)
		      (60:5) coordinate(4)
			  (0.833,4.33) coordinate(5)
			  (-0.833,4.33) coordinate(6)
			  (120:5) coordinate(7)
			  (-3.33,2.88) coordinate(8)
			  (-4.167,1.44) coordinate(9)
			  (180:5) coordinate(10)
			  (-4.167,-1.44) coordinate(11)
			  (-3.33,-2.88) coordinate(12)
			  (240:5) coordinate(13)
			  (-0.833,-4.33) coordinate(14)
			  (0.833,-4.33) coordinate(15)
			  (300:5) coordinate(16)
			  (3.33,-2.88) coordinate(17)
			  (4.167,-1.44) coordinate(18);
	
		\draw (1) -- (2) -- (3) -- (4) -- (5) -- (6) -- (7) -- (8) -- (9) -- (10) -- (11) -- (12) -- (13) -- (14) -- (15) -- (16) -- (17) -- (18) -- (1);

		\draw (2) -- (9);
		\draw (11) -- (18);
		\draw (3) -- (14);
		\draw (5) -- (12);
		\draw (6) -- (17);
		\draw (8) -- (15);

		\draw (2) arc (450:270:2 and 1.44);
		\draw[rotate=240] (3) arc (90:270:2 and 1.44);
		\draw[rotate=300] (6) arc (90:270:2 and 1.44);
		\draw (9) arc (90:270:2 and 1.44);
		\draw[rotate=60] (12) arc (90:270:2 and 1.44);
		\draw[rotate=120] (15) arc (90:270:2 and 1.44);

		\draw (1) coordinate[s1,fill,label=0:$P_b$];
		\draw (2) coordinate[c1,fill];
		\draw (3) coordinate[c1,fill];
		\draw (4) coordinate[c1,fill];
		\draw (5) coordinate[c1,fill];
		\draw (6) coordinate[c1,fill,label=90:$P_f$];
		\draw (7) coordinate[s1,fill,label=90:$P_d$];
		\draw (8) coordinate[c1,fill];
		\draw (9) coordinate[c1,fill=white,label=150:$P_c$];
		\draw (9) node[c1,fill=none] {\Huge $\times$};
		\draw (9) coordinate[c1,draw];
		\draw (10) coordinate[c1,fill];
		\draw (11) coordinate[c1,fill];
		\draw (12) coordinate[s1,fill,label=180:$P_a$];
		\draw (13) coordinate[c1,fill];
		\draw (14) coordinate[c1,fill];
		\draw (15) coordinate[s1,fill,label=270:$P_e$];
		\draw (16) coordinate[c1,fill];
		\draw (17) coordinate[c1,fill];
		\draw (18) coordinate[c1,fill];

		\coordinate[l1,label=260:$C_1$] (c1) at (10);
		\coordinate[l2,label=25:$C_2$] (c2) at (10);

		\end{tikzpicture}
		\fi
		\includegraphics{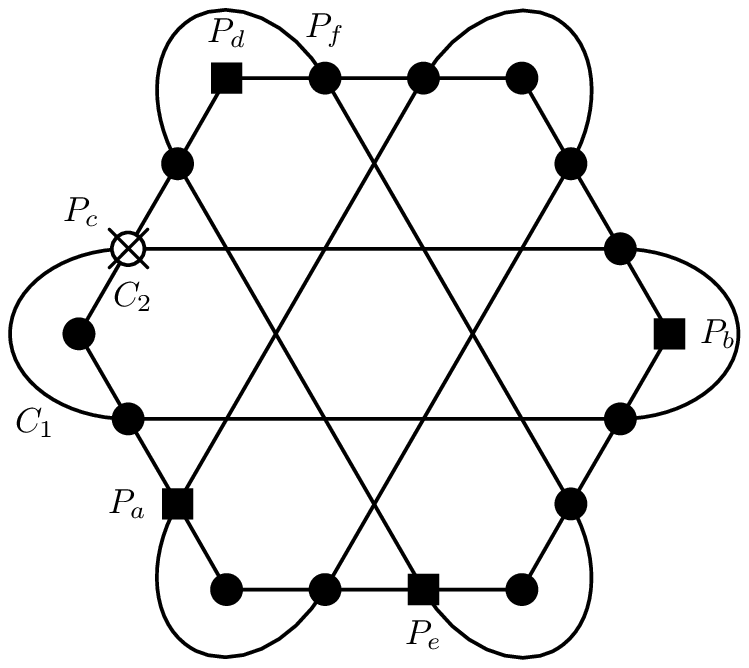}
		
		& \qquad\qquad\qquad &
		
		\if01
		\begin{tikzpicture}  [scale=0.6]

		\tikzstyle{every path}=[line width=1pt]
		\tikzstyle{c1}=[circle,minimum size=6]
		\tikzstyle{s1}=[rectangle,minimum size=9]
		\tikzstyle{l1}=[draw=none,circle,minimum size=35]
		\tikzstyle{l2}=[draw=none,circle,minimum size=12]

		\path (0:5) coordinate(1)
			  (4.167,1.44) coordinate(2)
			  (3.33,2.88) coordinate(3)
		      (60:5) coordinate(4)
			  (0.833,4.33) coordinate(5)
			  (-0.833,4.33) coordinate(6)
			  (120:5) coordinate(7)
			  (-3.33,2.88) coordinate(8)
			  (-4.167,1.44) coordinate(9)
			  (180:5) coordinate(10)
			  (-4.167,-1.44) coordinate(11)
			  (-3.33,-2.88) coordinate(12)
			  (240:5) coordinate(13)
			  (-0.833,-4.33) coordinate(14)
			  (0.833,-4.33) coordinate(15)
			  (300:5) coordinate(16)
			  (3.33,-2.88) coordinate(17)
			  (4.167,-1.44) coordinate(18);
	
		\draw (1) -- (2) -- (3) -- (4) -- (5) -- (6) -- (7) -- (8) -- (9) -- (10) -- (11) -- (12) -- (13) -- (14) -- (15) -- (16) -- (17) -- (18) -- (1);

		\draw (2) -- (9);
		\draw (11) -- (18);
		\draw (3) -- (14);
		\draw (5) -- (12);
		\draw (6) -- (17);
		\draw (8) -- (15);

		\draw (2) arc (450:270:2 and 1.44);
		\draw[rotate=240] (3) arc (90:270:2 and 1.44);
		\draw[rotate=300] (6) arc (90:270:2 and 1.44);
		\draw (9) arc (90:270:2 and 1.44);
		\draw[rotate=60] (12) arc (90:270:2 and 1.44);
		\draw[rotate=120] (15) arc (90:270:2 and 1.44);

		\draw (1) coordinate[c1,fill,label=0:$P_b$];
		\draw (2) coordinate[c1,fill=white];
		\draw (2) coordinate[c1,draw];
		\draw (3) coordinate[c1,fill];
		\draw (4) coordinate[c1,fill=white];
		\draw (4) coordinate[c1,draw];
		\draw (5) coordinate[c1,fill];
		\draw (6) coordinate[c1,fill=white];
		\draw (6) coordinate[c1,draw];
		\draw (7) coordinate[c1,fill=white];
		\draw (7) coordinate[c1,draw];
		\draw (8) coordinate[c1,fill=white];
		\draw (8) coordinate[c1,draw];
		\draw (9) coordinate[c1,fill=white,label=150:$P_c$];
		\draw (9) coordinate[c1,draw];
		\draw (10) coordinate[c1,fill];
		\draw (11) coordinate[c1,fill];
		\draw (12) coordinate[s1,fill,label=180:$P_a$];
		\draw (13) coordinate[c1,fill];
		\draw (14) coordinate[c1,fill];
		\draw (15) coordinate[c1,fill=white];
		\draw (15) coordinate[c1,draw];
		\draw (16) coordinate[c1,fill=white];
		\draw (16) coordinate[c1,draw];
		\draw (17) coordinate[c1,fill=white];
		\draw (17) coordinate[c1,draw];
		\draw (18) coordinate[c1,fill=white];
		\draw (18) coordinate[c1,draw];


		\end{tikzpicture}
		\fi
		\includegraphics{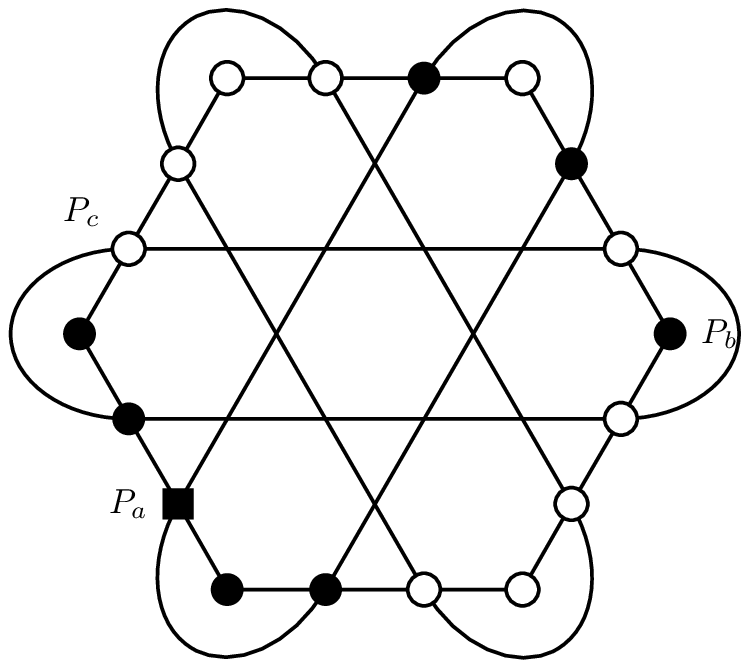}
		\\
		(a) & & (b)\\
	\end{tabular}
	\end{center}
	\caption{Greechie orthogonality diagram of a proof of the Kochen-Specker theorem~\cite{Cabello:1996zr}.
			The value of of $v$ of each observable (node) $P$ is represented as follows: $v(P)=1$ -- black square; $v(P)=0$ -- filled circle; $v(P)$ undefined (value indefinite) -- hollow circle.
			(a) The contradiction arising when $v(P_a)=v(P_b)=1$: $v$ cannot be admissible, since this would require that $v(P_c)=0$ and $v(P_c)=1$ simultaneously, as shown by the cross in the diagram.
			(b) A possible admissible value assignment when $v(P_a)=1$ and $v(P_b)=0$.}
	\label{fig:Cabello}
\end{figure}

Consider the Greechie orthogonality digram shown in Fig.~\ref{fig:Cabello}, in which vertices depict observables and smooth lines or curves represent contexts.
This well known diagram represents the `orthogonality' relations between the observables used in a well known proof of the Kochen-Specker theorem due to Cabello et al.~\cite{Cabello:1996zr}, containing only 18 one-dimensional projection observables on $\C^4$.

The Kochen-Specker theorem implies that there is no way to assign every observable in this diagram a value such that the admissibility requirements hold: exactly one observable in each context should have the value 1.

Let us suppose for the sake of example that $v(P_a)=v(P_b)=1$ and that $v$ is admissible.
Then, by working from $P_a$ and $P_b$ and applying the admissibility rule (a) one deduces that all observables in a context with $P_a$ or $P_b$ must take the value 0.
One then notices that there are contexts containing 3 observables with the value 0, so we can deduce from (b) that the fourth must have the value 1.
If we follow this line of reasoning, we can continue to assign values to observables with the admissibility requirements, as depicted in Fig.~\ref{fig:Cabello}(a), where a black square represents the value 1, and a black circle the value 0.
As we can see, by considering the contexts $C_1$ and $C_2$ we can infer that $P_c$ must take both the values 1 and 0 respectively: both possibilities contradict the admissibility of $v$, as does the final possibility -- that $P_c$ is value indefinite.
Note that, in Fig.~\ref{fig:Cabello}(a), the contradiction obtained at $P_c$ marked by the cross is a consequence of a specific succession of applications of the admissibility rules (a) and (b) in Definition~\ref{admis-rules}.
By applying these rules in a different order, one can obtain the contradiction also at $P_d$, $P_e$, or $P_f$.

The most important aspect of this reasoning in this context is that it is deterministic: we proceed only by deducing the value definiteness of observables via (a) and (b).

Now let us  assume that $v(P_a)=1$ and $v(P_b)=0$, as depicted in Fig.~\ref{fig:Cabello}(b).
We again apply (a) to observables commuting with $P_a$;
however, we then see that neither (a) nor (b) can be used again to deduce the value of another observable.
Normally, in proving that this diagram permits no consistent assignment of definite values, one would then proceed by assuming that one of the unfilled observables, such as $P_c$, must have either $v(P_c)=1$ or $v(P_c)=0$, and trying both possibilities.
One can do this when proving the Kochen-Specker theorem since one assumes (i): that every observable must have a definite value.
However, in order to localise value indefiniteness we do not make this assumption.
Hence, the value assignment in Fig.~\ref{fig:Cabello}(b), with the observables represented by unfilled circles being value indefinite (e.g., $v(P_c)$ undefined) represents an admissible value assignment.

Thus, under the assumption that $v(P_a)=1$, Fig.~\ref{fig:Cabello} does not suffice to prove that $v(P_b)$ must be value indefinite,  and hence cannot be used to localise value indefiniteness.
It is not difficult to see that we reach the same conclusion irrespective of our choice of observables as $P_a$ and $P_b$.

In this paper, in proving the main theorem, we give a set of observables for which this is the case.
That is, there are observable $P_a$ and $P_b$ such that if $v(P_a)=1$ then both $v(P_b)=0$ and $v(P_b)=1$ lead, via admissibility, to contradictions.

\section{The localised variant of the Kochen-Specker theorem}

Let us now state the strengthened theorem which is the focus of this paper.
As we mentioned, this generalises the results of~\cite{Abbott:2012fk,Abbott:2013ly} and uses a different proof technique allowing for a more symmetrised analytic approach.
The result in~\cite{Abbott:2013ly}, on the other hand, relies on computational results and the interpretation of graphs.

\begin{theorem}
	\label{thm:main}
    Let $n\ge 3$ and $\ket{\psi}, \ket{\phi}\in \C^n$ be unit vectors such that $0 < |\iprod{\psi}{\phi}| < 1$.
    We can effectively find a finite set of one-dimensional projection observables $\mathcal{O}$ containing $P_\psi$ and $P_\phi$
	for which there is no admissible value assignment function on $\mathcal{O}$ such that $v(P_\psi)=1$ and $P_\phi$ is value definite.
\end{theorem}

Before we proceed to prove Theorem~\ref{thm:main}, let us first discuss some important relevant issues.

This theorem has a slightly different form from the standard Kochen-Specker theorem because of the requirement that a particular observable in the set $\mathcal{O}$ be assigned the value 1.
However, since, as we will see, it is only the orthogonality relations between the observables in $\mathcal{O}$ which is important, a change of basis can always ensure that the required observable $P_\psi$ be assigned the value 1.

In order to interpret this result one has to take into account the eigenstate assumption discussed in the previous section: If a quantum system is prepared in a state $\ket{\psi}$ in $n\ge 3$ dimensional Hilbert space, then every one-dimensional projection observable that does not commute with $P_\psi$ is value indefinite and hence cannot have a predetermined measurement outcome.

\subsection{Insufficiency of existing Kochen-Specker diagrams}

The first question to address is whether existing Kochen-Specker diagrams (i.e., Greechie diagrams specifying the orthogonality relations of $\mathcal{O}$) could be used to provide a set $\mathcal{O}$ of observables proving Theorem~\ref{thm:main};
it is not \emph{a priori} obvious that such diagrams are unable to do so.
In Section~\ref{sec:CabelloExample} we showed, as an example, that a particular simple and well-known Kochen-Specker diagram is not sufficient for this purpose.
A careful search through existing diagrams showed that this is the case in general, and we were unable to find an existing Kochen-Specker diagram in which there are two observables $P_a$ and $P_b$ with the required property that if $v(P_a)=1$, both $v(P_b)=0$ and $v(P_b)=1$ lead to a contradiction.

A second conceptual problem with the use of fixed Kochen-Specker diagrams as in existing proofs is the following.
Since, in order to derive a contradiction, we need to assume that an observable $P_\psi$ in the given observable set has $v(P_\psi)=1$, this limits the observables which can be shown to be value indefinite to, at best, the remaining ones in $\mathcal{O}\setminus \{P_\psi\}$.
However, we wish to prove more: that \emph{every} one-dimensional projection observable not commuting with $P_\psi$ is value indefinite.

As a result, we need not only a set of observables with the required properties discussed above, but furthermore an approach to generalise this set of observables to arbitrary other observables.
We overcome this apparent lack of generality via a method of reductions, which we present in the next section and will return to discuss later on.

\subsection{Proof of Theorem~\ref{thm:main}}

We prove Theorem~\ref{thm:main} in three main steps:
\begin{enumerate}
	\item We first prove it for the special case that $|\iprod{\psi}{\phi}|=\frac{1}{\sqrt{2}}$.
		A similar result (for $|\iprod{\psi}{\phi}|=\frac{3}{\sqrt{14}}$) was shown in Ref.~\cite{Abbott:2012fk}, but this involved two separate diagrams applying to separate cases.
		Here we give a single diagram providing a much more compact, clear proof.
	\item We prove a simple reduction for $0 < |\iprod{\psi}{\phi}|<\frac{1}{\sqrt{2}}$ to the first case.
	\item The third and main part of the proof involves finding a reduction in the opposite sense, applying to the final $1 > |\iprod{\psi}{\phi}|>\frac{1}{\sqrt{2}}$ case.
		It is this final reduction allowing the complete proof that is the most involved technical aspect of this paper.
\end{enumerate}

As is standard in Kochen-Specker proofs~\cite{Cabello:1994ly}, we will work directly in the three-dimensional case of $\C^3$ (in fact, only $\R^3$ is needed), since the case for $n>3$ can be simply reduced to this situation.

\begin{lemma}\label{lemma:ExplicitCase}
	Given any two unit vectors $\ket{a},\ket{b}\in\C^3$ with $|\iprod{a}{b}|=\frac{1}{\sqrt{2}}$ there exists a finite set of one-dimensional projection observables $\mathcal{O}$ such that if $v(P_a)=1$ then $P_b$ is value indefinite under every admissible assignment function $v$ on $\mathcal{O}$.
\end{lemma}
\begin{proof}
	By choosing an appropriate basis we can assume, without loss of generality, that $\ket{a}=(1,0,0)$ and $\ket{b}=\frac{1}{2}(1,\sqrt{2},1)$.
	Let us consider the set $\mathcal{O}=\{P_a,P_b, P_i;\ i=1,\dots,35$\} of one-dimensional projection observables where the vectors $\ket{i}$ for $i=1,\dots,35$ are defined in Table~\ref{table:vectorList} (with the normalisation factors emitted for simplicity).
	The orthogonality relations between these vectors gives the 26 contexts shown in Table~\ref{table:contextList}.
	Note that these observables are `tightly' connected: the context-observable ratio is relatively high.
	The Greechie diagram showing the orthogonality relations is shown in Fig.~\ref{fig:greechie}.
	
	\begin{table}[ht]
		\caption{The 37 vectors specifying the observables used in the proof of Lemma~\ref{lemma:ExplicitCase}, with normalisation factors omitted.}
		\begin{tabular}{lllll}
		\hline
		$\ket{a}=(1,0,0)$ & $\ket{b} = (\sqrt{2},1,1)$ & $\ket{1}=(0,1,1)$ & $\ket{2} = (0,1,-1)$ & $\ket{3}=(\sqrt{2},-1,-1)$\\
		$\ket{4}=(0,0,1)$ & $\ket{5}=(0,1,0)$ & $\ket{6}=(\sqrt{2},1,-3)$ & $\ket{7}=(1,-\sqrt{2},0)$ & $\ket{8}=(\sqrt{2},-3,1)$\\
		$\ket{9}=(1,0,-\sqrt{2})$ & $\ket{10}=(\sqrt{2},1,0)$ & $\ket{11}=(\sqrt{2},0,1)$ & $\ket{12}=(\sqrt{2},-2,-3)$ & $\ket{13}=(1,-\sqrt{2},\sqrt{2})$\\
		$\ket{14}=(\sqrt{2},-3,-2)$ & $\ket{15}=(1,\sqrt{2},-\sqrt{2})$ & $\ket{16}=(\sqrt{8},1,-1)$ & $\ket{17}=(\sqrt{8},-1,1)$ & $\ket{18}=(\sqrt{2},-7,-3)$\\
		$\ket{19}=(\sqrt{2},-1,3)$ & $\ket{20}=(\sqrt{2},-3,-7)$ & $\ket{21}=(\sqrt{2},3,-1)$ & $\ket{22}=(1,\sqrt{2},0)$ & $\ket{23}=(1,0,\sqrt{2})$\\
		$\ket{24}=(\sqrt{2},-1,-3)$ & $\ket{25}=(\sqrt{2},-1,1)$ & $\ket{26}=(\sqrt{2},-3,-1)$ & $\ket{27}=(\sqrt{2},1,-1)$ & $\ket{28}=(\sqrt{2},-1,0)$\\
		$\ket{29}=(\sqrt{2},0,-1)$ & $\ket{30}=(\sqrt{2},2,3)$ & $\ket{31}=(\sqrt{2},3,2)$ & $\ket{32}=(\sqrt{2},3,7)$ & $\ket{33}=(\sqrt{2},7,3)$\\
		$\ket{34}=(\sqrt{2},1,3)$ & $\ket{35}=(\sqrt{2},3,1)$ & & &\\
		\hline
		\end{tabular}
		\label{table:vectorList}
	\end{table}
	
	\begin{table}[ht]
		\caption{The 26 contexts used in the proof of Lemma~\ref{lemma:ExplicitCase}.}
		\begin{tabular}{llllll}
		\hline
		$C_1=\{P_a,P_1,P_2\}$ & $C_2=\{P_a,P_4,P_5\}$ & $C_3=\{P_b,P_2,P_3\}$ & $C_4=\{P_b,P_6,P_7\}$ & $C_5=\{P_b,P_8,P_9\}$\\
		$C_6=\{P_4,P_7,P_{10}\}$ & $C_7=\{P_5,P_9,P_{11}\}$ & $C_8=\{P_{10},P_{12},P_{13}\}$ & $C_9=\{P_{11},P_{14},P_{15}\}$ & $C_{10}=\{P_1,P_{13},P_{16}\}$\\
		$C_{11}=\{P_1,P_{15},P_{17}\}$ & $C_{12}=\{P_{16},P_{18},P_{19}\}$ & $C_{13}=\{P_{17},P_{20},P_{21}\}$ & $C_{14}=\{P_3,P_{19},P_{22}\}$ & $C_{15}=\{P_3,P_{21},P_{23}\}$\\
		$C_{16}=\{P_{22},P_{24},P_{25}\}$ & $C_{17}=\{P_{23},P_{26},P_{27}\}$ & $C_{18}=\{P_4,P_{22},P_{28}\}$ & $C_{19}=\{P_5,P_{23},P_{29}\}$  & $C_{20}=\{P_{15},P_{28},P_{30}\}$\\
		$C_{21}=\{P_{13},P_{29},P_{31}\}$ & $C_{22}=\{P_8,P_{16},P_{32}\}$ & $C_{23}=\{P_6,P_{17},P_{33}\}$ & $C_{24}=\{P_7,P_{27},P_{34}\}$  & $C_{25}=\{P_9,P_{25},P_{35}\}$\\
		$C_{26}=\{P_1,P_{25},P_{27}\}$. & & & &\\
		\hline
		\end{tabular}
		\label{table:contextList}
	\end{table}
	
	\begin{figure}
		\begin{center}
		\includegraphics{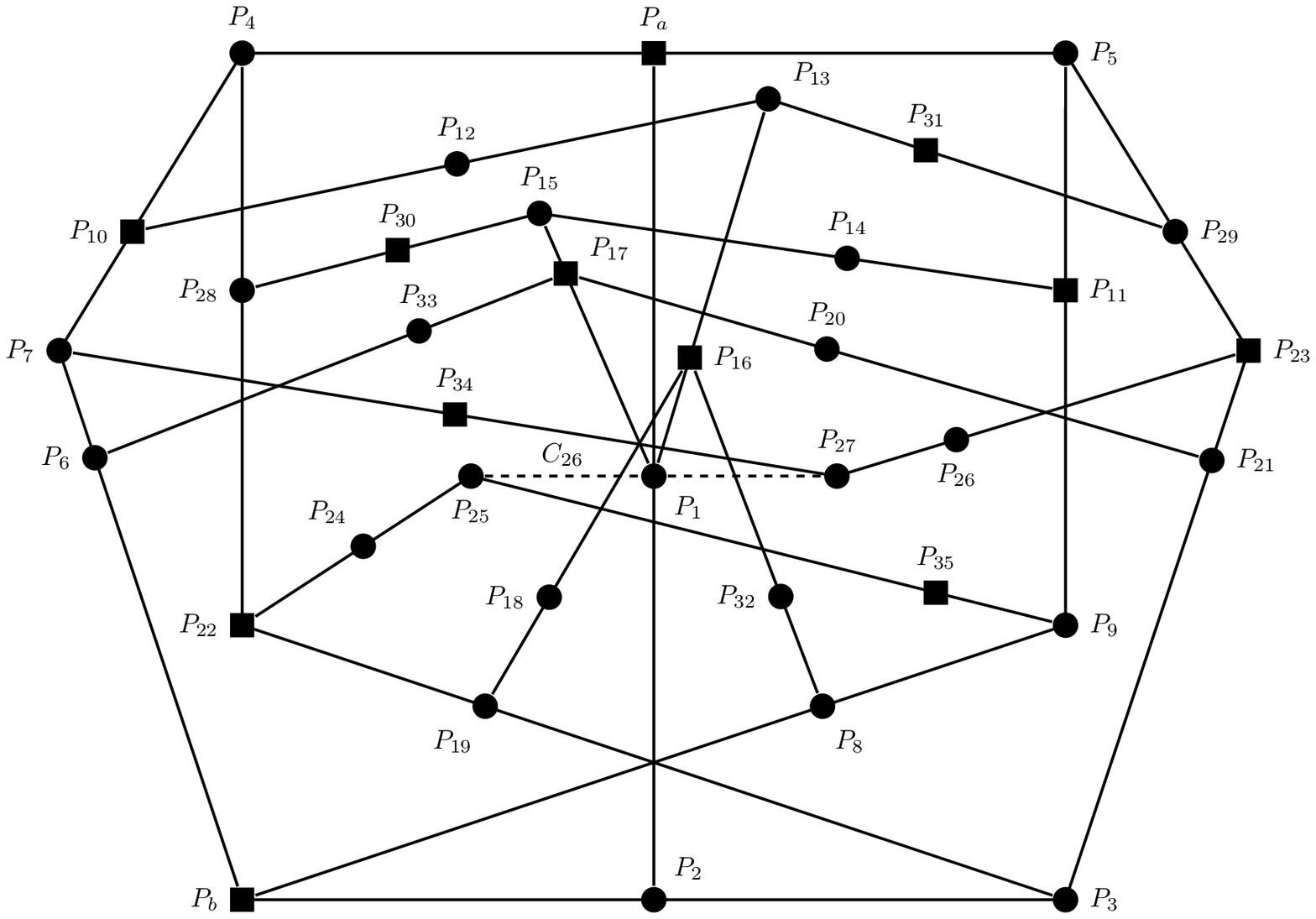}
		\if01
		\begin{tikzpicture}  [scale=0.6]

		\tikzstyle{every path}=[line width=1pt]
		\tikzstyle{c1}=[circle,minimum size=6]
		\tikzstyle{s1}=[rectangle,minimum size=9]
		\tikzstyle{l1}=[draw=none,circle,minimum size=4]


		\draw (4,0) coordinate[s1,fill,label=180:$P_b$] (b)
			-- (13,0) coordinate[c1,fill,label=45:$P_2$] (2)
			-- (22,0) coordinate[c1,fill,label=0:$P_3$] (3)
			-- (26,12) coordinate[c1,fill,pos=0.8,label=0:$P_{21}$] (21) coordinate[s1,fill,label=0:$P_{23}$] (23)
			-- (22,18.5) coordinate[c1,fill,pos=0.4,label=0:$P_{29}$] (29) coordinate[c1,fill,label=0:$P_5$] (5)
			-- (13,18.5) coordinate[s1,fill,label=90:$P_a$] (a)
			-- (4,18.5) coordinate[c1,fill,label=90:$P_4$] (4)
			-- (0,12) coordinate[s1,fill,pos=0.6,label=180:$P_{10}$] (10) coordinate[c1,fill,label=180:$P_7$] (7)
			-- (b) coordinate[c1,fill,pos=0.2,label=180:$P_6$] (6);
	
		\draw (a) -- (2) coordinate[c1,fill,pos=0.5,label=315:$P_1$] (1);

		\draw (5) -- (22,6) coordinate[s1,fill,pos=0.4,label=0:$P_{11}$] (11) coordinate[c1,fill,label=0:$P_9$] (9)
			-- (b) coordinate[c1,fill,pos=0.3,label=280:$P_8$] (8);
	
		\draw (4) -- (4,6) coordinate[c1,fill,pos=0.4,label=180:$P_{28}$] (28) coordinate[s1,fill,label=180:$P_{22}$] (22)
			-- (3) coordinate[c1,fill,pos=0.3,label=260:$P_{19}$] (19);

		\coordinate (25) at ([xshift=-4cm]1);
		\coordinate (27) at ([xshift=4cm]1);

		\draw (22) -- (25) coordinate[c1,fill,pos=0.5,label=115:$P_{24}$] (24) coordinate[c1,fill,label=270:$P_{25}$] (25)
			-- (9) coordinate[s1,fill,pos=0.8,label=90:$P_{35}$] (35);
	
		\draw (7) -- (27) coordinate[s1,fill,pos=0.5,label=90:$P_{34}$] (34) coordinate[c1,fill,label=90:$P_{27}$] (27)
			-- (23) coordinate[c1,fill,pos=0.3,label=270:$P_{26}$] (26);
	
		\draw (10) -- (15.5,17.5) coordinate[c1,fill,pos=0.5,label=90:$P_{12}$] (12) coordinate[c1,fill,label=15:$P_{13}$] (13)
			-- (29) coordinate[s1,fill,pos=0.4,label=90:$P_{31}$] (31);
	
		\draw (28) -- (10.5,15) coordinate[s1,fill,pos=0.5,label=90:$P_{30}$] (30) coordinate[c1,fill,label=90:$P_{15}$] (15)
			-- (11) coordinate[c1,fill,pos=0.6,label=90:$P_{14}$] (14);

		\draw (15) -- (1) coordinate[s1,fill,pos=0.2,label=15:$P_{17}$] (17)
			-- (13) coordinate[s1,fill,pos=0.3,label=0:$P_{16}$] (16);
	
		\draw (19) -- (16) coordinate[c1,fill,pos=0.3,label=180:$P_{18}$] (18)
			-- (8) coordinate[c1,fill,pos=0.7,label=180:$P_{32}$] (32);

		\draw (6) -- (17) coordinate[c1,fill,pos=0.7,label=90:$P_{33}$] (33)
			-- (21) coordinate[c1,fill,pos=0.4,label=90:$P_{20}$] (20);
	
		\draw[dashed] (25) -- (1) -- (27);
		
		\coordinate (ContextLabel) at ([shift=({-2cm,-3mm})]1);
		\draw (ContextLabel) coordinate[l1,label=90:$C_{26}$];

		\end{tikzpicture}
		\fi
		\end{center}
		\caption{Greechie diagram showing the orthogonality relation between the observables in Table~\ref{table:vectorList}. We have shown the deduction for $v(P_a)=v(P_b)=1$, where black squares represent the value 1, and circles the value 0. Observe that the context $C_{26}$, shown dotted, contains three observables with the value 0, and hence $v$ is not admissible.}
		\label{fig:greechie}
	\end{figure}
	
	Let us assume, for the sake of contradiction, than an admissible $v$ exists for $\mathcal{O}$, with $v(P_a)=1$ and $v(P_b)$ defined (i.e., $P_b$ value definite).
	Then there are two cases: $v(P_b)=1$ or $v(P_b)=0$.
	
	\emph{Case 1: $v(P_b)=1$.}
	Since $P_a\in C_1,C_2$ and $v(P_a)=1$, admissibility requires that $v(P_1)=v(P_2)=v(P_4)=v(P_5)=0$.
	Similarly, since $P_b\in C_3,C_4,C_5$ we have $v(P_3)=v(P_6)=v(P_7)=v(P_8)=v(P_9)=0$.
	Since $v(P_4)=v(P_7)=0$, admissibility in $C_6$ means that we must have $v(P_{10})=1$; similarly $v(P_{11})=1$ also.
	This chain of reasoning can be continued, applying the admissibility rules from Definition~\ref{admis-rules} one context at a time, as shown in Table~\ref{table:proofTable1}.
	In this table, where the leftmost column indicates the value of $v$ on the given observables, the values shown in bold in each column (context) are deduced from the admissibility rules based on the values of the other observables in the context which have already been deduced in the preceding columns.
	Note that, at each step, admissibility requires that certain observables take particular values;
	we never proceed by reasoning that $v(P_i)$ must be either 0 or 1 for some $P_i$ as is common in proofs of the standard Kochen-Specker theorem (except for $P_b$, where this is exactly the assumption that $P_b$ is value definite), because this is not required by admissibility.
	Eventually, as we see, we deduce that $v(P_1)=v(P_{25})=v(P_{27})=0$.
	But since $C_{26}=\{P_1,P_{25},P_{27}\}$, this contradicts the admissibility of $v$.
	
	\begin{table}[ht]
		\setlength{\tabcolsep}{6pt}
		\caption{The values that must be taken for the shown observables under any admissible assignment function $v$ satisfying $v(P_a) = v(P_b)=1$. The value (shown in the leftmost column) for observables in bold is deduced from the admissibility rules and observables appearing in columns to the left of that observable in the table.}
		\begin{tabular}{cccccccccccccccccccc}
		\hline
		$v$ && $C_1$ & $C_2$ & $C_3$ & $C_4$ & $C_5$ & $C_6$ & $C_7$ & $C_8$ & $C_9$ & $C_{10}$ & $C_{11}$ & $C_{12}$ & $C_{13}$ & $C_{14}$ & $C_{15}$ & $C_{16}$ & $C_{17}$\\
		\hline
		\hline
		$1$ && $P_a$ & $P_a$ & $P_b$ & $P_b$ & $P_b$ & $\bf P_{10}$ & $\bf P_{11}$ & $P_{10}$ & $P_{11}$ & $\bf P_{16}$ & $\bf P_{17}$ & $P_{16}$ & $P_{17}$ & $\bf P_{22}$ & $\bf P_{23}$ & $P_{22}$ & $P_{23}$\\
		$0$ && $\bf P_1$ & $\bf P_4$ & $P_2$ & $\bf P_6$ & $\bf P_8$ & $P_4$ & $P_5$ & $\bf P_{12}$ & $\bf P_{14}$ & $P_1$ & $P_1$ & $\bf P_{18}$ & $\bf P_{20}$ & $P_3$ & $P_3$ & $\bf P_{24}$ & $\bf P_{26}$\\
		$0$ && $\bf P_2$ & $\bf P_5$ & $\bf P_3$ & $\bf P_7$ & $\bf P_9$ & $P_7$ & $P_9$ & $\bf P_{13}$ & $\bf P_{15}$ & $P_{13}$ & $P_{15}$ & $\bf P_{19}$ & $\bf P_{21}$ & $P_{19}$ & $P_{21}$ & $\bf P_{25}$ & $\bf P_{27}$\\
		\hline
		\end{tabular}
		\label{table:proofTable1}
	\end{table}
	
	\emph{Case 2: $v(P_b)=0$.}
	By following a similar line of reasoning, shown in Table~\ref{table:proofTable2}, we once again deduce that $v(P_1)=v(P_{25})=v(P_{27})=0$, a contradiction.

	\begin{table}[ht]
		\setlength{\tabcolsep}{6pt}
		\caption{The values that must be taken for the shown observables under any admissible assignment function $v$ satisfying $v(P_a)=1$ and $v(P_b)=0$. As in Table~\ref{table:proofTable1}, the bold values represent the observables with values deduced from previous observables in the table.}
		\begin{tabular}{cccccccccccccccccccc}
		\hline
		$v$ && $C_1$ & $C_2$ & $C_3$ & $C_{14}$ & $C_{15}$ & $C_{18}$ & $C_{19}$ & $C_{20}$ & $C_{21}$ & $C_{10}$ & $C_{11}$ & $C_{22}$ & $C_{23}$ & $C_4$ & $C_5$ & $C_{24}$ & $C_{25}$\\
		\hline
		\hline
		$1$ && $P_a$ & $P_a$ & $\bf P_3$ & $P_3$ & $P_3$ & $\bf P_{28}$ & $\bf P_{29}$ & $P_{28}$ & $P_{29}$ & $\bf P_{16}$ & $\bf P_{17}$ & $P_{16}$ & $P_{17}$ & $\bf P_7$ & $\bf P_9$ & $P_7$ & $P_9$\\
		$0$ && $\bf P_1$ & $\bf P_4$ & $P_b$ & $\bf P_{19}$ & $\bf P_{21}$ & $P_4$ & $P_5$ & $\bf P_{15}$ & $\bf P_{13}$ & $P_1$ & $P_1$ & $\bf P_8$ & $\bf P_6$ & $P_b$ & $P_b$ & $\bf P_{27}$ & $\bf P_{25}$\\
		$0$ && $\bf P_2$ & $\bf P_5$ & $P_2$ & $\bf P_{22}$ & $\bf P_{23}$ & $P_{22}$ & $P_{23}$ & $\bf P_{30}$ & $\bf P_{31}$ & $P_{13}$ & $P_{15}$ & $\bf P_{32}$ & $\bf P_{33}$ & $P_6$ & $P_8$ & $\bf P_{34}$ & $\bf P_{35}$\\
		\hline
		\end{tabular}
		\label{table:proofTable2}
	\end{table}
	
	Hence, we must conclude that $P_b$ cannot be value definite if $v$ is admissible on $\mathcal{O}$.
\end{proof}

We next show a `contraction' lemma that constitutes a simple `forcing' of value definiteness: given $P_a$ and $P_b$ with $v(P_a)=v(P_b)=1$, there is a $\ket{c}$ which is `closer' (i.e.,~at a smaller angle of our choosing; contracted) to both $\ket{a}$ and $\ket{b}$, for which $v(P_c)=1$ also.
This result was proved in~\cite{Abbott:2012fk}, but we reproduce the short proof here for completeness.
The form of the vectors $\ket{c_\pm}$ specified in the lemma will be used several times in the rest of the paper.
\begin{lemma}[Contraction Lemma, \cite{Abbott:2012fk}]
	\label{lemma:reduction1}
	 Given any two unit vectors $\ket{a},\ket{b}\in\C^3$ with $0 < |\iprod{a}{b}| < 1$ and a $z\in\C$ such that $|\iprod{a}{b}| < |z| < 1$, we can effectively find a unit vector $\ket{c}$ with $
	\iprod{a}{c}=z$,  and a finite set of one-dimensional projection observables $\mathcal{O}$ containing $P_a$, $P_b$, $P_c$
	 such that if $v(P_a)=v(P_b)=1$, then $v(P_c)=1$, for
every admissible assignment function  $v$  on $\mathcal{O}$.
	
	Furthermore, if we choose our basis such that $\ket{a}= (0,0,1)$ and $\ket{b}=(\sqrt{1-|p|^2},0,p)$, where $p=\iprod{a}{b}$, then $\ket{c}$ can only be one of the following two vectors: $\ket{c_\pm}=(x,\pm y,z)$, where $z=\iprod{a}{c}$, $x=p(1-z^2)/(z\sqrt{1-p^2})$ and $y=\sqrt{1-x^2-z^2}$.
\end{lemma}
\begin{proof}
	Without loss of generality, we assume the $\iprod{a}{b}\in \R$ and choose a basis so that $\ket{a}=(0,0,1)$ and $\ket{b}=(q,0,p)$ where $p=\iprod{a}{b}$ and $q=\sqrt{1-p^2}$.
	
	Note that, since $p<|z|$ and thus $p^2<z^2$ we have
	$$\frac{p^2(1-z^2)}{q^2z^2} = \frac{p^2-p^2z^2}{q^2z^2} < \frac{z^2-p^2z^2}{q^2z^2} = \frac{(1-p^2)z^2}{q^2z^2} = 1.$$
	If we let $x=\frac{p(1-z^2)}{qz}$ we thus have
	$$x^2 = \frac{p^2(1-z^2)}{q^2z^2}(1-z^2) < 1-z^2.$$
	We can then set $y = \sqrt{1-x^2-z^2}\in\R$, making $\ket{c}=(x,y,z)$ a unit vector such that $\iprod{a}{c}=z$.
	
	Let $\ket{\alpha}=\ket{a}\times\ket{c}=(-y,x,0)$, $\ket{\beta}=\ket{b}\times\ket{c}=(-py,px-qz,qy)$ and note that $\iprod{\alpha}{\beta}=0$ also.
	Thus, if we let $\ket{\alpha'}=\ket{a}\times\ket{\alpha}$ and $\ket{\beta'}=\ket{b}\times\ket{\beta'}$, then  $\{\ket{a},\ket{\alpha},\ket{\alpha'}\}$, $\{\ket{b},\ket{\beta},\ket{\beta'}\}$ and $\{\ket{\alpha},\ket{\beta},\ket{c}\}$ are all orthonormal bases for $\R^3$ and thus $C_1=\{P_\alpha,P_\beta,P_c\}$, $C_2=\{P_a,P_\alpha,P_{\alpha'}\}$ and $C_3=\{P_b,P_\beta,P_{\beta'}\}$ are all contexts in $\mathcal{O}=C_1\cup C_2 \cup C_3$.
	This construction is illustrated in Fig.~\ref{fig:reduction1}.
	
	If $v$ is an admissible assignment function on $\mathcal{O}$ with $v(P_a)=v(P_b)=1$ then we must have $v(P_\alpha)=v(P_\beta)=0$ and hence $v(P_c)=1$, as required.
\end{proof}

\begin{figure}[ht]
\begin{center}
	\includegraphics{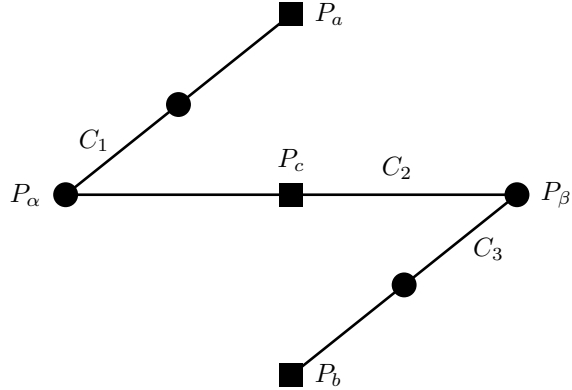}
	\if01
	\begin{tikzpicture}  [scale=0.6]

	\tikzstyle{every path}=[line width=1pt]
	\tikzstyle{c1}=[circle,minimum size=6]
	\tikzstyle{s1}=[rectangle,minimum size=9]
	\tikzstyle{l1}=[draw=none,circle,minimum size=4]

	\draw (5,0) coordinate[s1,fill,label=0:$P_b$] (b)
		-- (10,4) coordinate[c1,fill,pos=0.5] (c3) coordinate[l1,pos=0.7,label=0:$C_3$] coordinate[c1,fill,label=0:$P_\beta$] (beta)
		-- (0,4) coordinate[s1,fill,pos=0.5,label=90:$P_c$] (c) coordinate[l1,pos=0.35,label=25:$C_2$] coordinate[c1,fill,label=180:$P_\alpha$] (alpha)
		-- (5,8) coordinate[c1,fill,pos=0.5] (c2) coordinate[l1,pos=0.3,label=180:$C_1$] coordinate[s1,fill,label=0:$P_a$] (a);

	\end{tikzpicture}
	\fi
\end{center}
\caption{Greechie orthogonality diagram with an overlaid value assignment that illustrates the reduction in Lemma~\ref{lemma:reduction1}.
         Once again, the circles and squares represent observables that have the values $0$ and $1$ respectively.}
\label{fig:reduction1}
\end{figure}

We now present a proof for the reduction in the opposite direction: finding (from $\ket{a},\ket{b}$) two vectors $\ket{c},\ket{d}$ specifying observables $P_c,P_d$ for which $v(P_c)=v(P_d)=1$, and which are further apart from each other than $\ket{a}$ is from $\ket{b}$.
This is made easier by noting that it is not necessary to find a vector $\ket{c}$ `further' from $\ket{a}$ than $\ket{b}$, but rather just two vectors further from each other than $\ket{a}$ is from $\ket{b}$.

This process is broken into two steps.
We first prove an `Expansion Lemma' which, unlike the Contraction Lemma, does not find two vectors arbitrarily far apart satisfying the required criteria.
Rather, we then show a further lemma, the `Iteration Lemma', proving that this expansion can be iterated to meet the required conditions.

\begin{lemma}[Expansion Lemma]
	\label{lemma:reduction2}
	Given any two unit vectors $\ket{a},\ket{b}\in\C^3$ with $\frac{1}{3} < |\iprod{a}{b}|<1$, we can effectively find unit vectors $\ket{c},\ket{d}$ with $0<|\iprod{c}{d}|<|\iprod{a}{b}|$ and a finite set of one-dimensional projection observables $\mathcal{O}$ containing $P_a,P_b,P_c,P_d$ such that if $v(P_a)=v(P_b)=1$, then $v(P_{c})=v(P_{d})=1$, for
every admissible assignment function  $v$  on $\mathcal{O}$.
\end{lemma}
\begin{proof}
	Let $\iprod{a}{b}=\alpha$.
	Without loss of generality, we will consider only the positive, real case of $\frac{1}{3}<\alpha <1$.
	We fix an orthonormal basis such that, written in this basis, $\ket{a}$ and $\ket{b}$ lie in the $xz$-plane bisected by the $z$-axis.
	In this basis we thus have
	$$\ket{a}=\left(\sqrt{1-\beta^2},0,\beta\right),\ \ket{b}=\left(-\sqrt{1-\beta^2},0,\beta\right),$$
	where
	\begin{equation}\label{eqn:beta}
		\beta=\sqrt{\frac{\alpha+1}{2}}\,\cdot
	\end{equation}
	It is readily confirmed that
	$$\iprod{a}{b}=\beta^2-(1-\beta^2)=2\beta^2-1=\alpha$$
	as desired.
	Note that
	we thus have
	\begin{equation}\label{eq:betaBound}\sqrt{\frac{2}{3}}<\beta<1.\end{equation}

	Figure~\ref{fig:lemma2sphere} shows the contour representing all the possible vectors specifying observables which can be forced to take the value 1 from the construction in Lemma~\ref{lemma:reduction1}.
	We use two applications of Lemma~\ref{lemma:reduction1} applied to $\ket{a},\ket{b}$ to give two such vectors $\ket{c},\ket{d}$ lying in the $yz$-plane. 
	
	\begin{figure}[ht]
	\begin{center}
		\includegraphics{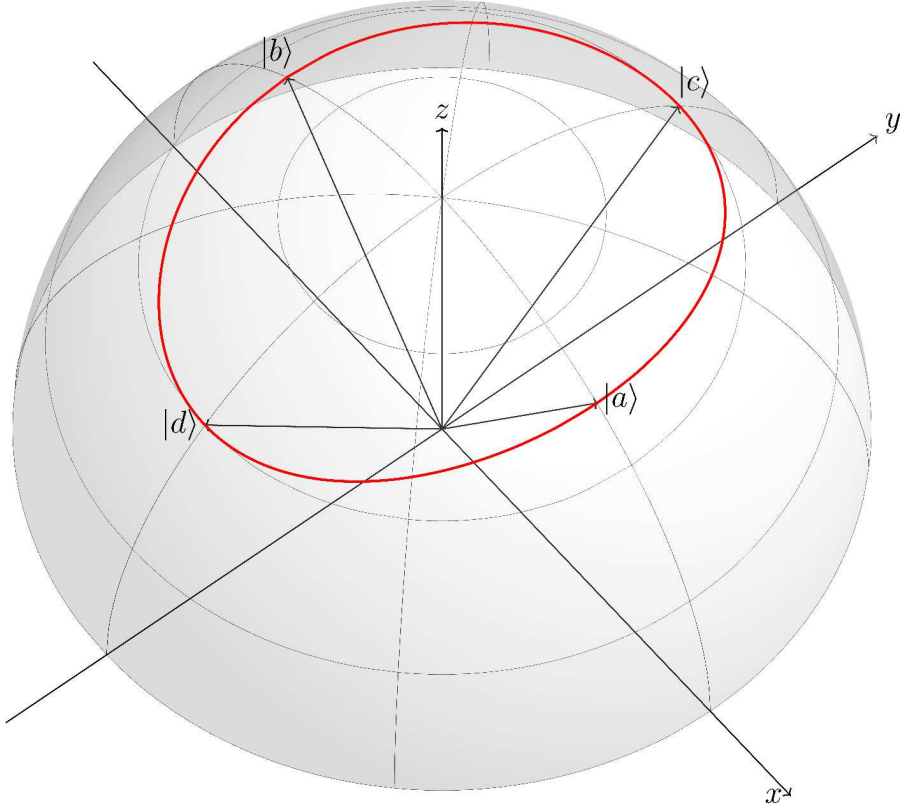}
	\end{center}
	\caption{A plot of the possible vectors $\ket{c}$ corresponding to the one-dimensional projection observables that Lemma~\ref{lemma:reduction1} can force to take the value 1. The bold (red; colour online) curve represents the position on the unit sphere of such vectors for given $\ket{a},\ket{b}$. Note the $\ket{c}$ and $\ket{d}$ are further apart from each other than $\ket{a}$ and $\ket{b}$.}
	\label{fig:lemma2sphere}
	\end{figure}
	
	We can also see, at least for the chosen values of $\ket{a},\ket{b}$ that are shown in Fig.~\ref{fig:lemma2sphere}, that $\iprod{a}{b}>\iprod{c}{d}$.
	Indeed it appears that the vectors `$\ket{c}$', `$\ket{d}$' shown in the $yz$-plane provide the maximum separation, and the symmetry under exchange of $\ket{a}$ and $\ket{b}$ of Lemma~\ref{lemma:reduction1}  seems to support this.
	However, it is not necessary to prove this is the case.
	Rather, we will show directly that the vectors $\ket{c},\ket{d}$ provide the required expansion.
	To do so, we derive a simple explicit form for $\ket{c},\ket{d}$ and thus $\iprod{c}{d}$.
	We focus first on finding $\ket{c}$; the form of $\ket{d}$ follows immediately.
	
	Rather than use basis-transformations to attempt to apply Lemma~\ref{lemma:reduction1} to find the form of $\ket{c},\ket{d}$ in this specific case, we will re-derive the result explicitly making use of our symmetrised basis choice.
	
	The vectors $\ket{a},\ket{b},\ket{c}$ need to follow the orthogonality relations shown in Fig.~\ref{fig:reduction1} in order to conclude that $v(P_{c})=1$.
	 That is, we need vectors $\ket{e},\ket{f}$ such that $\{\ket{e},\ket{f},\ket{c}\}$ is an orthonormal set, and further that $\iprod{a}{e}=\iprod{b}{f}=0$.
	
	Since we choose $\ket{c}$ to be in the $yz$-plane, we can write it in the parameterised form
	$\ket{c}=\left(0,\sqrt{1-\gamma^2},\gamma\right)$, where $\gamma >0$ remains to be found.
	Since $\ket{e}$ should be orthogonal to both $\ket{a}$ and $\ket{c}$, we have  $$\ket{e}=\ket{a}\times\ket{c}=\left(-\beta\sqrt{1-\gamma^2},-\gamma\sqrt{1-\beta^2},\sqrt{(1-\beta^2)(1-\gamma^2)}\right).$$
	Similarly, we have $$\ket{f}=\ket{b}\times\ket{c}=\left(-\beta\sqrt{1-\gamma^2},\gamma\sqrt{1-\beta^2},-\sqrt{(1-\beta^2)(1-\gamma^2)}\right).$$
	Further, the orthogonality of $\ket{e}$ and $\ket{f}$ gives us
	\begin{align*}
		 \iprod{e}{f}&=\beta^2(1-\gamma^2)-\gamma^2(1-\beta^2)-(1-\beta^2)(1-\gamma^2)\\
		 &=\beta^2 - \beta^2\gamma^2 - \gamma^2 + \beta^2\gamma^2 - 1 + \gamma^2 + \beta^2 - \beta^2\gamma^2\notag\\
		 &=2\beta^2-\beta^2\gamma^2-1\notag\\
		 &=0\notag
	\end{align*}
	and hence
	$\beta^2(2-\gamma^2)=1.$
	Thus,
	\begin{equation}\label{eqn:gamma}
		\gamma=\sqrt{2-\frac{1}{\beta^2}}\,\cdot
	\end{equation}
	Further, it is readily verified that $\frac{1}{\sqrt{2}}<\gamma<1$ for $\sqrt{\frac{2}{3}}<\beta<1,$ and hence for all $\frac{1}{3}<\alpha<1$ (recall Eqn.~\ref{eq:betaBound}).
	
	Similarly, we find $\ket{d}=(0,-\sqrt{1-\gamma^2},\gamma)$ using a further two auxiliary vectors $\ket{g},\ket{h}$ forming the orthonormal set $\{\ket{d},\ket{g},\ket{h}\}$ where $\iprod{a}{g}=\iprod{b}{h}=0$.
	
	Thus, if we take $\mathcal{O}=\{P_a,P_b,P_c,P_d,P_e,P_f,P_g,P_h\}$, as a result of the orthogonality relationships expressed in Fig.~\ref{fig:reduction1}, $v(P_c)=v(P_d)=1$ for any admissible $v$ on $\mathcal{O}$ with $v(P_a)=v(P_b)=1$.
	
	It remains then just to show that
	\begin{equation}\label{eqn:gleb}
		\iprod{c}{d}=2\gamma^2-1 < \iprod{a}{b}=\alpha=2\beta^2-1.
	\end{equation}
	We note that $\iprod{c}{d}>0$ for $\gamma>\frac{1}{\sqrt{2}}$.
	
	We finish the proof by showing proving Eqn.~\ref{eqn:gleb}, that is, that $\iprod{c}{d} < \alpha$, or, equivalently, $\gamma^2 < \beta^2$.
	But since we can write
	$$\left(\beta-\frac{1}{\beta}\right)^2=\beta^2-\frac{1}{\beta^2}-2$$
	we have from Eqn.~\ref{eqn:gamma}
	$$\gamma^2=2-\frac{1}{\beta^2}=\beta^2-\left(\beta-\frac{1}{\beta}\right)^2<\beta^2,$$
	 concluding the proof.
	
	 We note for completeness that we can write $\iprod{c}{d}$ directly in terms of $\alpha$ from Eqns.~\ref{eqn:beta}, \ref{eqn:gamma} and \ref{eqn:gleb} as
	 \begin{equation}\label{eqn:iprodfn}
	 	\iprod{c}{d} = 3-\frac{4}{\alpha+1}\,\cdot
	 \end{equation}
\end{proof}

We now prove that by iterating this procedure we can find a pair of vectors arbitrarily far apart from each other.

\begin{lemma} [Iteration  Lemma]
	\label{lemma:reduction3}
	
	Given any two unit vectors $\ket{a},\ket{b}\in\C^3$ with $\frac{1}{3}< |\iprod{a}{b}|<1$, we can effectively find unit vectors $\ket{c},\ket{d}$ with $0<|\iprod{c}{d}|\le \frac{1}{3}$ and a finite set of one-dimensional projection observables $\mathcal{O}$ containing $P_a,P_b,P_{c},P_{d}$ such that if $v(P_a)=v(P_b)=1$, then $v(P_{c})=v(P_{d})=1$, for
every admissible assignment function  $v$  on $\mathcal{O}$.
\end{lemma}
\begin{proof}
	We prove by iterating Lemma~\ref{lemma:reduction2}, and use the notation $\ket{c_0}\equiv \ket{a}$ and $\ket{d_0}\equiv \ket{b}$, indicating the 0th iteration.
	We start with $\ket{c_0},\ket{d_0}$ and for each $i\ge 0$, as long as $\ket{c_{i}},\ket{d_{i}}$ satisfy $\iprod{c_{i}}{d_{i}}> \frac{1}{3}$, apply the construction used in the proof of Lemma~\ref{lemma:reduction2} to generate $\ket{c_{i+1}},\ket{d_{i+1}}$ for the next iteration. In  particular, $\ket{c_{i+1}},\ket{d_{i+1}}$  satisfy the equality~\eqref{eqn:iprodfn} for  $\alpha_i=\iprod{c_i}{d_i}$ (in particular, $\alpha_0=\iprod{c_0}{d_0}=\iprod{a}{b}$).

	By Lemma~\ref{lemma:reduction2}, we know that $\iprod{c_i}{d_i}>\iprod{c_{i+1}}{d_{i+1}}$ for each iteration $i$.
We now prove that	
	the process cannot produce an infinite sequence $\ket{c_0},\ket{d_0};\ket{c_1},\ket{d_1};\cdots$, with $\iprod{c_i}{d_i}>\frac{1}{3}$ for all $i$, that is, for
some  $i$ we have $\iprod{c_i}{d_i}\le\frac{1}{3}$.
	(The sequence must stop here, since Lemma \ref{lemma:reduction2} cannot be applied for $\iprod{c_i}{d_i}\le \frac{1}{3}$.)
	
	From Eqn.~\ref{eqn:iprodfn} we define the function $s:\left(\frac{1}{3},1\right)\to (0,1)$ such that $$s(u)=3-\frac{4}{u+1}\,\raisebox{.8mm}{,}$$ giving the inner product of the next pair in the iteration.
	We thus have $s(\alpha_0)=\alpha_1$ and, more generally, $\alpha_i=s^i(\alpha_0)$.
	We can thus rephrase the problem: \emph{does there exist a $k$ such that $s^k(\alpha_0)\le\frac{1}{3}$?}
	
	Let us, for the sake of contradiction, assume the contrary.
	Then $(\alpha_i)_i=(s^i(\alpha_0))_i$ is an infinite strictly decreasing sequence of reals with $\alpha_i>\frac{1}{3}$ for all $i$.
	For any finite $i$ we thus have
	\begin{align*}
		s^i(\alpha_0)=\alpha_i &= \alpha_0 - |\alpha_1 - \alpha_0| - \cdots - |\alpha_i - \alpha_{i-1}|\notag\\
		&= \alpha_0 - (\alpha_0 - \alpha_1) - \cdots - (\alpha_{i-1}-\alpha_i)\notag\\
		&= \alpha_0 - \sum_{k=0}^{i-1}(\alpha_k - \alpha_{k+1}).		
	\end{align*}
	Let us define the function $D:\left(\frac{1}{3},1\right)\to \left(0,\frac{1}{3}\right)$ such that $$D(u) = u - s(u) = u - \left(3-\frac{4}{u+1}\right)$$ so that
	$$
		\alpha_i = \alpha_0 - \sum_{k=0}^{i-1}D(\alpha_k).
	$$

	We can show that $\frac{\dd D}{\dd u} < 0$ for $u\in\left(\frac{1}{3},1\right)$: calculating the derivative we have
		$$\frac{\dd D}{\dd u} = 1 - \frac{4}{(u+1)^2}
			< 1 - \frac{4}{(1+1)^2}= 0.$$
	Since $D$ is thus a strictly decreasing function
	on $\left(\frac{1}{3},1\right)$ and $\alpha_k<\alpha_0$ for all $k>0$,
	we have $D(\alpha_0)< D(\alpha_k)$ for all $k>0$.
	Hence we set
	$$\alpha_i = \alpha_0 - \sum_{k=0}^{i-1}D(\alpha_k) < \alpha_0 - iD(\alpha_0).$$
	Since $D(\alpha_0)=\alpha_0 - \alpha_1 > 0$ is a positive constant, it is not possible that $s^i(\alpha_0) = \alpha_i > \frac{1}{3}$, for all $i>0$, because in this case we would have $\frac{1}{3} < \alpha_{0}-iD(\alpha_{0})$, for all $i>0$, a contradiction.

	In fact, if $k$ is the smallest positive integer greater than $\frac{\alpha_{0}-\frac{1}{3}}  {D(\alpha_{0})}$, then $\alpha_k\le   \frac{1}{3}$, as required.
	We note that $s^{k+1}(\alpha_0)$ is not defined.
	
	By Lemma~\ref{lemma:reduction2}, for each $i=0,\dots,k-1$ there exists a set $\mathcal{O}_i$ of one-dimensional projection observables such that $v(P_{c_{i+1}})=v(P_{d_{i+1}})=1$ under any $v$ admissible on $\mathcal{O}_i$ satisfying $v(P_{c_i})=v(P_{d_i})=1$.
	Hence, if we take the set $\mathcal{O}=\cup_{i=0}^{k-1}\mathcal{O}_i$ we must have $v(P_{c_k})=v(P_{d_k})=1$ under any admissible $v$ on $\mathcal{O}$ satisfying $v(P_a)=v(P_b)=1$, and $\iprod{c_k}{d_k}\le \frac{1}{3}$, as required.
\end{proof}

With these lemmata proved, we are in a position to combine them to prove Theorem~\ref{thm:main}.

\begin{proof}[Proof of Theorem~\ref{thm:main}]	
	If we have $|\iprod{\psi}{\phi}|=\frac{1}{\sqrt{2}}$ then, by Lemma~\ref{lemma:ExplicitCase}, there exists a finite set $\mathcal{O}$ of one-dimensional projection observables for which there is no admissible $v$ on $\mathcal{O}$ satisfying the requirements, so we are done.
	
	Otherwise, we proceed directly to prove that if $\mathcal{O}$ is a set of one-dimensional projection observables containing $P_\psi,P_\phi$ then no admissible assignment function $v$ on $\mathcal{O}$ with $v(P_\psi)=1$ can have $P_\phi$ value definite.
	We show this in two cases: first that $v(P_\phi)\neq 1$ and then that $v(P_\phi)\neq 0$.
	Let us first show that there is a set $\mathcal{O}_1$ for which $v(P_\phi)\neq 1$ if $v$ is admissible on $\mathcal{O}_1$.
	
	There are two cases: either $0 < |\iprod{\psi}{\phi}| < \frac{1}{\sqrt{2}}$ or $1 > |\iprod{\psi}{\phi}| > \frac{1}{\sqrt{2}}$.
	
	If $0 < |\iprod{\psi}{\phi}| < \frac{1}{\sqrt{2}}$, then by Lemma~\ref{lemma:reduction1} there exists a vector $\ket{\phi'}$ such that $\iprod{\psi}{\phi'}=\frac{1}{\sqrt{2}}$ and a set $\mathcal{O}_2$ of observables containing $P_\psi,P_\phi,P_{\phi'}$ such that if $v$ is admissible on $\mathcal{O}_2$, $v(P_{\phi'})=1$ also.
	But, by Lemma~\ref{lemma:ExplicitCase}, there exists a set $\mathcal{O}_3$ of one-dimensional projection observables containing $P_\psi,P_{\phi'}$ such that if $v$ is admissible on $\mathcal{O}_3$ and $v(P_\psi)=1$, $P_{\phi'}$ must be value indefinite.
	Thus, if we take $\mathcal{O}_1=\mathcal{O}_2\cup \mathcal{O}_3$ we cannot have $v(P_\phi)=1$ as required.
	
	If $1 > |\iprod{\psi}{\phi}| > \frac{1}{\sqrt{2}}$, then by Lemma~\ref{lemma:reduction3} there exist two vectors $\ket{\psi'},\ket{\phi'}$ such that $0<|\iprod{\psi'}{\phi'}|\le \frac{1}{3}$ and a set $\mathcal{O}_4$ of observables containing $P_\psi,P_\phi,P_{\psi'},P_{\phi'}$ such that if $v$ is admissible on $\mathcal{O}_4$ then $v(P_{\psi'})=v(P_{\phi'})=1$ also.
	But, by Lemma~\ref{lemma:reduction1}, there exists a vector $\ket{\phi''}$ such that $\iprod{\psi'}{\phi''}=\frac{1}{\sqrt{2}}$ and a set $\mathcal{O}_5$ of observables containing $P_{\psi'},P_{\phi''},P_{\phi'}$ such that if $v$ is admissible, $v(P_{\phi''})=1$ also.
	Finally, once more by Lemma~\ref{lemma:ExplicitCase}, there exists a set $\mathcal{O}_6$ for which $v$ there is no admissible $v$ on $\mathcal{O}_5$ satisfying $v(P_{\psi'})=v(P_{\phi''})=1$.
	Hence, there is no admissible $v$ on the set $\mathcal{O}_1=\mathcal{O}_4\cup \mathcal{O}_5\cup\mathcal{O}_6$ such that $v(P_\phi)=1$ as required.
	
	This shows that there exists a set $\mathcal{O}_1$ of one-dimensional projection observables containing $P_\psi,P_\phi$ such that we cannot have $v(P_\phi)=1$ if $v(P_\psi)=1$ if $v$ is admissible $\mathcal{O}_1$.
	It remains to show that there exists a set $\mathcal{O}_0$ such that we cannot have $v(P_\phi)=0$ if $v$ is admissible on $\mathcal{O}_0$.
	
	Let us assume, without loss of generality, that $\ket{\psi}=(1,0,0)$ and $\ket{\phi}=(p,\sqrt{1-p^2},0)$ where $p=|\iprod{\psi}{\phi}|$.
	Then let $\ket{\alpha}=(0,1,0)$, $\ket{\beta}=(0,0,1)$ and $\ket{\phi'}=(\sqrt{1-p^2},p,0)$.
	Then $\{\ket{\psi},\ket{\alpha},\ket{\beta}\}$ and $\{\ket{\phi},\ket{\phi'},\ket{\beta}\}$ are orthonormal bases for $\C^3$ and hence $C_1=\{P_\psi,P_\alpha,P_\beta\}$ and $C_2=\{P_\phi,P_{\phi'},P_\beta\}$ are contexts in $\mathcal{O}_7=C_1 \cup C_2$.
	But if $v$ is admissible on $\mathcal{O}_7$ and $v(P_\psi)=1$, $v(P_\phi)=0$, admissibility implies that $v(P_\phi)=1$.
	
	As we have shown just before, there exists a set $\mathcal{O}_8$ of one-dimensional projection observables containing $P_\psi,P_{\phi'}$ such that there is no admissible assignment function $v$ on $\mathcal{O}_8$ with $v(P_\psi)=v(P_{\phi'})=1$, and hence there is no admissible $v$ on $\mathcal{O}_0=\mathcal{O}_7\cup \mathcal{O}_8 $ such that $v(P_\psi)=1$ and $v(P_\phi)=0$.
	
	Having covered all cases, we are forced to conclude that there is a set $\mathcal{O}=\mathcal{O}_0\cup \mathcal{O}_1$ of observables containing $P_\psi$ and $P_\phi$ such that if $v(P_\psi)=1$, $P_\phi$ cannot be value definite if $v$ is admissible on $\mathcal{O}$.
\end{proof}

\section{Discussion}

The important difference between Theorem~\ref{thm:main} and the Kochen-Specker theorem lies in what physical conclusions can be drawn from the theorems which, of course, are purely mathematical results.
Key to interpreting such theorems is the recognition that a value assignment represents a possible hidden variable assignment for a quantum system, and that the value assigned to an observable thus represents the result that would be obtained upon its measurement.
Under this interpretation the Kochen-Specker theorem shows that, given a system prepared in the quantum state $\ket{\psi}$ in dimension 3 or higher Hilbert space, the results of all possible measurements on the state $\ket{\psi}$ cannot be predetermined (noncontextually) as they would in a classical theory.
It says nothing, however, about whether all, or simply a few, outcomes are not predetermined. On the other hand, 
Theorem~\ref{thm:main} implies that no one-dimensional projection observable $P$
can have a predetermined measurement outcome for the system unless 
$\ket{\psi}$ is an eigenstate of $P$.
This interpretation relies on the eigenstate assumption discussed earlier in the paper, stating that the observable $P_\psi$ has a predetermined measurement outcome -- a very weak assumption.
Conceptually, this means that Theorem~\ref{thm:main} goes further than the Kochen-Specker theorem in showing the \emph{extent} of non-classicality that the quantum logic event-structure implies.

It is possible to generalise this result -- that formally applies only to one-dimensional projection observables --
to the value-indefiniteness of more general classes of observables.
Since an observable $A$ (formally a Hermitian operator in $n$-dimensional Hilbert space)
with a non-degenerate spectrum, distinct eigenvalues $a_1,\dots,a_n$ and eigenstates $\ket{a_1},\dots,\ket{a_n}$
can be expressed as its spectral decomposition
$A=\sum_{i=1}^n a_i P_{a_i}$ (where $P_{a_i}=\frac{\oprod{a_i}{a_i}}{|\iprod{a_i}{a_i}|}$, as usual),
it physically has a predetermined measurement outcome if and only if all the projectors $P_{a_i}$, $i=1,\dots, n$, have predetermined measurement outcomes\footnote{Specifically, if one such $P_{a_i}$ has the predetermined value $1$ then one must obtain $a_i$ upon measurement of $A$; admissibility then requires that all $P_{a_j}$ have the definite value 0 for $j\neq i$.} -- that is, are value definite.
Thus, for a system prepared in a state $\ket{\psi}$ in dimension 3 or higher Hilbert space, the outcome of a measurement of an observable $A$ with non-degenerate spectra cannot be predetermined (noncontextually) unless $\ket{\psi}$ is an eigenstate of $A$.


\subsection{Proof size}

Since the first appearance of the Kochen-Specker theorem~\cite{Kochen:1967fk}, much attention has been given to reducing the number of observables and contexts needed to obtain a contradiction and prove the theorem.
The original result used a set of 117 observables, but more recent results have, to quote some notable examples, shown sets containing 31 observables (in three-dimensional Hilbert space)~\cite{Peres:1991ys} and 18 observables (in four-dimensional Hilbert space)~\cite{Cabello:1996zr}.

While such results do not affect the interpretation of the theorem, they have merit in showing the depth of the contradiction between the classical and quantum logical structures.
More recently, smaller proofs have been of particular interest since these have been used to derive noncontextuality inequalities that can be experimentally tested~\cite{Cabello:2008hc} in the same vein as Bell-inequalities~\cite{Bell:1966uq};
smaller sets of observables lead to smaller and more readily testable inequalities.

Conceptually, however, the key point is probably that the theorem can be proved using a \emph{finite} set of observables;
if a contradiction only arose when an infinity of observables were considered, this would potentially raise questions about the constructive and operational character of the theorem and its use of counterfactuals, hence its interpretation would be more questionable~\cite{Pitowsky:1998aa}.

The localised nature of Theorem~\ref{thm:main} immediately means that a single finite set $\mathcal{O}$ of one-dimensional projection observables will never suffice to prove the value indefiniteness of all such projection observables $P_\phi$ not commuting with $P_\psi$ for a given state $\ket{\psi}$.
There are infinitely many such observables, and one must, by definition, include $P_\phi$ in $\mathcal{O}$ to localise value indefiniteness to $P_\phi$.
Rather, the nature of Theorem~\ref{thm:main} means we must look for constructive methods to obtain a set $\mathcal{O}_\phi$ for a given $P_\phi$, which is precisely what we have done in our proof of the result.

Of course, a given set of orthogonality relations (i.e., a Greechie diagram) may be realisable for an infinity of different such sets $\mathcal{O}$, as is the case with the diagram depicted in Fig.~\ref{fig:reduction1}.
Thus, it would be preferable to find a given set of orthogonality relations for which a set $\mathcal{O}_\phi$ of observables realising these relations and containing both $P_\psi$ and $P_\phi$ for any $P_\phi$.
Since we were unable to give such a set of relations, we had to iterate Lemma~\ref{lemma:reduction2} a number of times times depending on $P_\psi$, with no upper bound (but only ever finitely many times).

Furthermore, it seems that it is difficult, if not impossible, to succeed in giving a fixed set of orthogonality relations that works in all cases.
In order to show an observable $P_a$ has $v(P_a)=1$ using the admissibility requirements, one must give a context $\{P_a,P_b,P_c\}\subset \mathcal{O}$ for which it is already known that $v(P_b)=v(P_c)=0$.
This implies two observable $P_d$ and $P_e$ such that $v(P_d)=v(P_e)=1$ and $\iprod{b}{d}=\iprod{c}{e}=0$.
But this is precisely the case described in Lemma~\ref{lemma:reduction1}.
However, in Lemma~\ref{lemma:reduction2} we showed the limitations of this process in `widening the angle' between vectors whose corresponding projectors both take the value 1 -- hence the necessity of iterating Lemma~\ref{lemma:reduction2}.

As a result it seems that, in contrast to the Kochen-Specker theorem, arbitrarily large (but always finite) sets of observables are needed to show that a given observable $P_\phi$ is value indefinite.
Nonetheless, the critical point is once again that for \emph{any given} $P_\phi$, we can show that $P_\phi$ is value indefinite with a \emph{finite} set of observables, and hence that the counterfactual reasoning used is no more problematic than in the Kochen-Specker theorem.

\subsection{State-independence and testability}

One of the strengths of the Kochen-Specker theorem that has been repeatedly emphasised is the fact that the contradiction between its hypotheses is derived independently of the state a quantum system is prepared in; this is commonly referred to as state-independence.
This is in contrast to violation of Bell-type inequalities (which occur only for particular entangled states) and shows that the non-classicality  results from the structure of quantum mechanics itself, rather than features of particular states, such as entanglement~\cite{Kirchmair:2009gr,Zu:2012tj}.
Consequently, various experimental inequalities based on the Kochen-Specker theorem that, although often simpler, are state-dependent have been criticised, and much effort has been expended to find simple, state-independent inequalities to test~\cite{Cabello:2008hc}.

In contrast to the Kochen-Specker theorem, the form of Theorem~\ref{thm:main} and, in particular, the interpretation (relying, of course, on the eigenstate assumption) that \emph{for a system prepared in a given state $\ket{\psi}$}, any one-dimensional projection observable $P_\phi$ not commuting with $P_\psi$ is value indefinite, may suggest that Theorem~\ref{thm:main} does not share this state-independence.
As a result, this issue deserves a little discussion.

The state-independence of the Kochen-Specker theorem ensures that no quantum state in $n\ge 3$ dimensional Hilbert space admits a classical assignment of definite values to all observables within certain finite sets.
This is true also with Theorem~\ref{thm:main}: for any quantum state $\ket{\psi}$, all projection observables not contained within the `star' of one-dimensional
projection observables commuting with $P_\psi$ (see Fig.~\ref{fig:star}) are value indefinite.

\begin{figure}
	\begin{center}
		\includegraphics{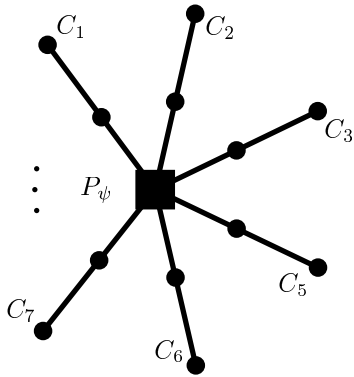}
		\if01
		\begin{tikzpicture}  [scale=0.8]
		\tikzstyle{every path}=[line width=1pt]
		\tikzstyle{every node}=[draw,line width=1pt,inner sep=0]

		\tikzstyle{c1}=[rectangle,minimum size=6]

		\tikzstyle{d1}=[circle,draw=none,fill,minimum size=2]

		\tikzstyle{l7}=[draw=none,circle,minimum size=45]

		\draw[black,line width=2pt] (0:0) -- (135:3)
		        coordinate[c1,minimum size=14,fill,at start] (0)
		        coordinate[c1,circle,midway,fill] (1)
		        coordinate[c1,circle,at end,fill,label=35:$C_1$] (2);

		\draw[black,line width=2pt] (0.center) -- (90:3)
		        coordinate[c1,at start] (0)
		        coordinate[c1,circle,midway,fill] (3)
		        coordinate[c1,circle,at end,fill,label=350:$C_2$] (4);

		\draw[black,line width=2pt] (0.center) -- (45:3)
		        coordinate[c1,at start] (0)
		        coordinate[c1,circle,midway,fill] (5)
		        coordinate[c1,circle,at end,fill,label=305:$C_3$] (6);

		\draw[black,line width=2pt] (0.center) -- (315:3)
		        coordinate[c1,at start] (0)
		        coordinate[c1,circle,midway,fill] (9)
		        coordinate[c1,circle,at end,fill,label=215:$C_5$] (10);

		\draw[black,line width=2pt] (0.center) -- (270:3)
		        coordinate[c1,at start] (0)
		        coordinate[c1,circle,midway,fill] (11)
		        coordinate[c1,circle,at end,fill,label=170:$C_6$] (12);

		\draw[black,line width=2pt] (0.center) -- (225:3)
		        coordinate[c1,at start] (0)
		        coordinate[c1,circle,midway,fill] (13)
		        coordinate[c1,circle,at end,fill,label=125:$C_7$] (14);

		\draw[black,line width=2pt] (0.center) -- (0:3)
		        coordinate[c1,fill,at start] (0)
		        coordinate[c1,circle,fill,midway,fill] (7)
		        coordinate[c1,circle,,fill,at end,label=260:$C_4$] (8);

		\coordinate[l7,label=180:$P_\psi$] (0) at (0.center);

		\coordinate[d1,black] (.) at (190:2);
		\coordinate[d1,black] (.) at (180:2);
		\coordinate[d1,black] (.) at (170:2);
		\end{tikzpicture}
		\fi
	\end{center}
	\caption{Greechie diagram showing an observable $P_\psi$ with $v(P_\psi)=1$ and the (infinite) set of compatible observables $P_\phi$ for which $v(P_\phi)=0$. This is the maximal extent of value definiteness for a system in state $\ket{\psi}$ -- no other one-dimensional projection observables on $\C^3$ can be value definite.}
	\label{fig:star}
\end{figure}

Rather, it is not Theorem~\ref{thm:main} that is state-dependent, but the proof we have given: to show that a given observable $P_\phi$ is value indefinite from the assumption that $v(P_\psi)=1$, we need a set $\mathcal{O}$ particular to this $\ket{\phi}$.
However, as we discussed in the preceding section, this is perfectly reasonable given the form of the theorem.

One can emphasise further the state-independence of Theorem~\ref{thm:main} by restating the theorem in the following form: \emph{``Only a single one-dimensional projection observable on the Hilbert space $\C^n$ for $n\ge 3$ can be assigned the value 1 by an admissible, noncontextual value assignment function''}.
In this form the state-independence is clear;
the illusion of state-dependence enters because of the connection, via the eigenstate assumption, between the ``one observable assigned the value 1'' and the particular state $\ket{\psi}$ (and corresponding observable $P_\psi$ with $v(P_\psi)=1$) which is necessary for the physical interpretation of the theorem.


The importance of the state-independence of the Kochen-Specker theorem arises, in part, in the use of Kochen-Specker sets of observables in testable inequalities.
It is important to note that, even though these inequalities are sometimes referred to as ``Kochen-Specker inequalities''~\cite{Larsson:2002aa}, they are better seen simply as noncontextuality inequalities.
These inequalities are derived under the assumption only of noncontextuality, ignoring the admissibility requirements, and bounds on quantities are calculated over all possible noncontextual value assignments.
A key result shows that one can derive such an inequality from any Kochen-Specker set~\cite{Yu:2014aa}.
It is clear that these value assignments cannot obey the admissibility requirements, since the Kochen-Specker theorem shows precisely that no classical value assignment can do so.

The strength of Theorem~\ref{thm:main}, on the other hand, relies precisely on the use of the admissibility requirements to determine when definite values should be assigned.
Hence, while one can use the methods of~\cite{Yu:2014aa} to derive inequalities from the constructions in the proof of Theorem~\ref{thm:main}, these bounds would be calculated over all noncontextual value assignments (subject to $v(P_\psi)=1$), without paying heed to admissibility, and hence would offer no conceptual advantage over existing inequalities.
Furthermore, since our construction in Lemma~\ref{lemma:ExplicitCase}, for example, contains 37 observables, these would pose no experimental benefit to existing, simpler inequalities either~\cite{Kirchmair:2009gr}.

Nonetheless, the state-independence of the result shows that  the value indefiniteness of almost all one-dimensional projection observables in quantum mechanics is indeed a deep feature of the theory -- of the logical structure of Hilbert space -- rather than a property of particular states.

\section{Conclusions}

In summary, we proved a variant of the Kochen-Specker theorem showing that the non-classicality implied by the Kochen-Specker theorem is, in a specific sense, maximal.
Specifically, under the assumptions that (1) any value definite observables behave noncontextually, and (2) contexts obey weak `admissibility' rules on any value definite observables they contain, we show that only one one-dimensional projection observable on the Hilbert space $\C^n$ (for $n\ge 3$) can be assigned consistently the definite value 1.

If a quantum system is prepared in a state $\ket{\psi}$, subject to the assumption that any value assignment function for the system must assign the value 1 to $P_\psi$ since $P_\psi\ket{\psi}=\ket{\psi}$, the theorem can be interpreted as showing that the measurement of any observable $P_\phi$ projecting onto the linear subspace spanned by a state $\ket{\phi}$ that is neither orthogonal nor co-aligned with $\ket{\psi}$ must be value indefinite -- that is, indeterministic.
This interpretation, which shows that almost all one-dimensional projection observables are value indefinite for a given system~\cite{Abbott:2013ly}, is stronger than what can be drawn from the Kochen-Specker theorem, which, in contrast, shows only that not all observables can be value definite.

This result justifies further the general belief that quantum mechanics is indeterministic -- that there is no hidden variable or definite value determining the outcome of a measurement in advance.
This eliminates the need to assume that the non-classicality shown by the Kochen-Specker theorem should apply uniformly, instead deriving this global value indefiniteness.
As with the Kochen-Specker theorem, this result relies on the assumption that classical values, should they exist, must behave noncontextually.

Finally, these results help theoretically certify quantum random number generators~\cite{Abbott:2012fk}, since the promises of such devices rely on the indeterministic nature of quantum measurements~\cite{Stefanov:2000aa}.
By localising value indefinite observables, one can be sure that the measurements producing the output bits do not yield any pre-existing element of physical reality.
We emphasise that these results do not hold for two-dimensional systems -- a class into which many current quantum random number generators unfortunately fall.

\begin{acknowledgments}
We thank the anonymous referees for suggestions which helped improve this paper.
This work was supported in part by Marie Curie FP7-PEOPLE-2010-IRSES Grant RANPHYS.
\end{acknowledgments}


%

\end{document}